\documentclass{amsart}

\usepackage{cite}
\usepackage{latexsym}
\usepackage{amsmath, amssymb,amsfonts,amscd}
\usepackage{mathrsfs}
\usepackage{graphicx}
\usepackage{url}
\usepackage{enumerate}

\newtheorem{theorem}{Theorem}[section]
\newtheorem{lemma}[theorem]{Lemma}

\newtheorem{remark}[theorem]{\bf Remark}
\newtheorem{definition}[theorem]{Definition}

\newcommand{\RR}{\mathbb{R}}
\newcommand{\PP}{\mathbb{P}}

\newcommand{\AF}{\mathbb{A}}

\newcommand{\mb}[1][]{\mathbf}
\newcommand{\m}[1]{\mb{m_{#1}}}
\newcommand{\x}{\mb{x}}

\newcommand{\bs}[1][]{\boldsymbol}
\newcommand{\bt}{\boldsymbol{\tau}}
\newcommand{\eqn}{\begin{eqnarray}}
\newcommand{\feqn}{\end{eqnarray}}

\begin{document}
%\setcounter{page}{1001}
%\issue{XXI~(2001)}

\title[The bifurcation curve]{TDOA--based localization in two dimensions:\\
the bifurcation curve}

\author[M. Compagnoni, R. Notari]{Marco Compagnoni, Roberto Notari}
\address{Dipartimento di Matematica,
Politecnico di Milano, I-20133 Milano, Italia}
\email{marco.compagnoni@polimi.it\\
roberto.notari@polimi.it}

\maketitle

%\runninghead{M. Compagnoni, R. Notari}{The bifurcation curve}

\maketitle

\begin{abstract}
In this paper, we complete the study of the geometry of the TDOA
map that encodes the noiseless model for the localization of a
source from the range differences between three receivers in a
plane, by computing the Cartesian equation of the bifurcation
curve in terms of the positions of the receivers. From that
equation, we can compute its real asymptotic lines. The present manuscript completes the analysis of \cite{nostro}. Our result is useful to check if a source belongs or is closed to the bifurcation curve, where the localization in a noisy scenario is ambiguous.
\end{abstract}

%%%%%%%%%%%%%%%%%%%%%%%%%%%%%%

\section{Introduction}

The problem of localizing an object in space is a classical topic
in the literature of space--time signal processing. The first
studies on the subject date back to World War II, motivating the creation of
the two-dimensional LOng RAnge Navigation (LORAN) radio
positioning system. LORAN was based on the measurements of the time
differences of arrival (TDOA) of synchronized radio signals
originated from three distinct known emitters and it required the
use of hyperbolic charts to determine the position of the receiver
\cite{Getting1993}.

Nowadays, LORAN has been replaced primarily by Global Positioning
System (GPS), but the mathematical models underlying LORAN and
(three--dimensional) GPS localization are essentially the same
\cite{Siouris1993}. TDOA--based localization of unknown point
sources is very widespread and popular also in acoustic signal
processing, because it is characterized by a reduced computational
cost with respect to other solutions and robustness against noise
\cite{Huang2004}.

The main contributions to the study of TDOA--based localization
come from the engineering literature, where the authors usually
focus on the development of algorithms for locating the source
starting from empirical TDOA data, affected by (mainly, Gaussian
distributed) noise. Relevant examples are
\cite{Smith1987,Smith1987a,Abel1987,Schau1987,Huang2000,Huang2001,
Huang2004,Huang2004a,Gillette2008a,Beck2008}. A classification of
the different methods can be done according to the proposed
solution: maximum likelihood principle versus least--squares
estimators, linear approximation versus numerical optimization,
and finally iterative versus closed forms--algorithms (for a
resume of the most significant of these methods see
\cite{Bestagini2013}).

However, all of them are based on the model of geometric
propagation of the signal in an isotropic and homogeneous medium.
This means that, given a TDOA measurement between two receivers
placed on the Euclidean plane at positions $\mathbf{m}_i$ and
$\mathbf{m}_j$, the locus of source locations that are compatible
with that measurement is one branch of a hyperbola of foci
$\mathbf{m}_i$ and $\mathbf{m}_j$ and whose aperture depends on
the range difference (TDOA times the speed of propagation of the
signal). Therefore, if we consider multiple measurements, one can
readily find the location of the source through the intersection
of branches of hyperbole. The three dimensional localization is
very similar, being equivalent to the intersection of
sheets of hyperboloids.

The mathematical properties of the TDOA--based localization have
been investigated in several manuscripts. In particular, the
closed--form solutions to the intersection problems have been provided
for both configurations of three receivers in a plane and four
receivers in the space
\cite{Schmidt1972,Bancroft1985,Kraus1987,Abel1991,Chauffe1994,Leva1995,
Hoshen1996,Grafarend2002,Awange2002,Coll2009,Spencer2010,Coll2012,nostro}.
However, due to the nonlinearity, it is well known that in these
minimal configurations of receivers there does not exist a unique
admissible position of the source for any given set of TDOA
measurements. In fact, there are open regions in the physical
space where the intersection set is the union of two points and so
the source location is intrinsically ambiguous. This fact is known
in literature as the bifurcation problem and we will name
\emph{bifurcation set} the border between the two domains where the localization problem has respectively one or two solutions.

In \cite{Schmidt1972} it is shown that the bifurcation set is a
curve in two dimensions and a surface in three dimensions. In
\cite{Hoshen1996}, by relating TDOA--based localization and the
ancient Problem of Apollonius \cite{Boyer1989} of drawing a circle
touching three other circles or two circles and a point, Hoshen
was able to analytically describe the bifurcation sets in two and three dimension
in terms of polar and spherical coordinates. More recently, an
analysis of the bifurcation problem in a 2--D scenario has been completed through the use of numerical simulations \cite{Spencer2007}. Source location error analysis in 2--D and 3--D with noisy TDOA data has also recently been considered for both closed form and numerical solutions \cite{Spencer2010}.

A deeper description of the geometric properties of the
bifurcation curve in the 2--D case was given lately by the authors
in \cite{nostro}. Although the statistical model describing
TDOA--based localization is not defined by polynomial functions,
in \cite{nostro} it has been shown that all the relevant objects
are real (semi)--algebraic varieties. In particular, the
bifurcation curve is an algebraic curve and more precisely a
rational quintic curve, whose explicit rational parametrization
was provided. In \cite{nostro}, we were not able to find a
Cartesian polynomial equation.

In this manuscript, we fill this gap by providing the explicit
Cartesian equation of the bifurcation curve in terms of the
positions of the receivers. This is an important step towards
the study of the statistical model of localization based on
TDOAs. Indeed, as observed in \cite{nostro}, locating a source placed around the bifurcation curve is an ill-posed problem and in these situations any localization algorithm has a very low accuracy. On the other hand, checking if a point belongs to, or is close to, a curve is quite a hard problem if one starts from the parametric equation of the curve itself. On the contrary, the two above--mentioned problems can be easily solved from the Cartesian equation of the curve.

The paper is organized as follows. In Section 2, we describe the
deterministic model of the physical problem, and we recall the
main results proved in \cite{nostro}. Section 3 is devoted to the
computation of the Cartesian equation of the bifurcation curve and
of its real asymptotic lines in terms of the positions of the
receivers. In the last Section, we summarize the results, and we
outline a research program for extending our study to both the
real scenario, and the cases not yet covered, such as the 3--D
case, or the planar case with more than $ 3 $ receivers.

\section{State of the art}

In this section we briefly resume the main results of
\cite{nostro} concerning the 2--D localization with three
receivers in a noiseless scenario. In this setting, the physical
world can be identified with the Euclidean plane and, after
choosing an orthogonal Cartesian coordinate system, with $\RR^2$.
On the plane, we have three receivers $\m{i}=(x_{i},y_{i}),\
i=0,1,2$ at known positions and a source $\mb{x} = (x,y)$.
\begin{figure}[htb]
\begin{center}
\resizebox{5.7cm}{!}{
  \includegraphics%[bb=206 313 405 478]
  {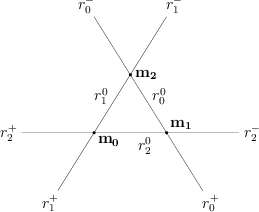}}
  \caption{\label{rette-acustica}Receivers $\m{0},\m{1},\m{2}$ in a generic planar configuration. The line $r_i$ does not contain the receiver $\m{i}$.}
\end{center}
\end{figure}
The corresponding displacement vectors are
\begin{equation}
\mb{d_i}(\mb{x})=\mb{x}-\m{i},\qquad
\mb{d_{ji}}=\m{j}-\m{i},\qquad i,j=0,1,2,
\end{equation}
whose norms are $d_i(\mb{x})$ and $ d_{ji}$, respectively.
Without loss of generality, we assume the speed of propagation of
the signal in the medium to be equal to $1$. In the noiseless
scenario we adopt, the TDOA between each pair of different
microphones is equal to the difference of the ranges (the
so--called pseudorange):
\begin{equation}\label{TDOA}
\tau_{ji}(\mb{x})=d_{j}(\mb{x})-d_{i}(\mb{x}),\quad
i,j=0,1,2.
\end{equation}

The three TDOAs are not independent. In fact, the linear relation
$ \tau_{12}(\x ) = \tau_{10}(\x ) - \tau_{20}(\x ) $ holds for
each $ \mb{x} \in \RR^2.$ Hence, we are allowed to choose a
microphone as the reference one, say $ \m{0},$ and so we consider only
the TDOAs involving $ \m{0},$ without loss of information.  In
\cite{nostro} we collected $ \tau_{10}(\x )$ and $\tau_{20}(\x )$
by defining the TDOA map:
\begin{equation}
\begin{array}{cccc}
\bs{\tau_2}: & \RR^2          & \longrightarrow & \RR^2\\
 & \mb{x}   & \longrightarrow & \quad
(\tau_{10}(\mb{x}),\tau_{20}(\mb{x}))
\end{array}.
\end{equation}

The study of the TDOA map is the heart of the mathematical
characterization of the localization problem and it is the subject
of \cite{nostro}. In fact, the main problems concerning the
localization in the deterministic set--up can be formulated in
terms of $ \bt_2 $: given $ \bt:=(\tau_1,\tau_2)\in\RR^2,$ there
exists a source $\mb{x}$ such that $ \bt_2(\mb{x}) = \bt$ if, and
only if, $ \bt \in \mbox{Im}(\bt_2).$ Moreover, the uniqueness of
such a source $\mb{x}$ can be equivalently written as $ \vert
\bt_2^{-1}(\bt) \vert = 1$.

In this paper we assume that $\m{0},\m{1},\m{2}$ are not
collinear. The interested reader can find the complete analysis of
the aligned configuration in \cite{nostro}. In Figure
\ref{fig:tauimage} we draw the image of $\bt_2$, with receivers
$\m{0}=(0,0),\ \m{1}=(2,0)$ and $\m{2}=(2,2)$.
\begin{figure}[htb]
\begin{center}
\resizebox{5cm}{!}{
  \includegraphics%[bb=227 296 384 495]
  {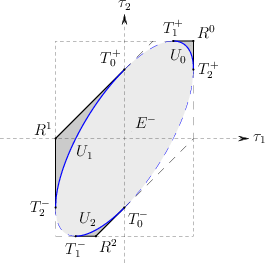}}
  \caption{\label{fig:tauimage}
  The image of $ \bs{\tau_2}$ is the gray subset of the hexagon with continuous and dashed sides. In the light gray region $E^-$ the map $ \bs{\tau_2} $ is $ 1$--to--$1,$ while in the medium gray region $U_0 \cup U_1 \cup U_2$ the map $ \bs{\tau_2} $ is $2$--to--$1.$ The vertices $R^i$ and the continuous part of the border of the hexagon and the blue ellipse $E$ are in the image, and there $ \bs{\tau_2} $ is $ 1$--to--$1.$ The points $T_i^\pm$ and the dashed part of the border of the hexagon and $E$ do not belong to Im($\bs{\tau_2}$).}
\end{center}
\end{figure}

\noindent Im$(\bs{\tau_2})$ is a subset of a convex polytope
$P_2$, the hexagon defined by:
\begin{equation}
\left\{ \begin{array}{l}
-d_{10} \leq \tau_1 \leq d_{10} \\
-d_{20} \leq \tau_2 \leq d_{20} \\
-d_{21} \leq \tau_2 - \tau_1 \leq d_{21}
\end{array} \right. .
\end{equation}
There exists a unique ellipse $E$ tangent to each facet of $P_2$
(at the six points $T_i^\pm,\ i=0,1,2$), the one defined by
\begin{equation}\label{eq:a}
a(\bs{\tau}) = \Vert \tau_2 \mb{d_{10}} - \tau_1 \mb{d_{20}} \Vert^2 -
\Vert \mb{d_{10}} \wedge \mb{d_{20}} \Vert^2=0.
\end{equation}
We name $E^-$ and $E^+$ the interior and the exterior region of
the ellipse, respectively, while $ U_0, U_1, U_2$ are the three
disjoint connected components of $\mathring{P_2} \setminus
(E^-\cup E)$. Using this notation, we have
\begin{equation}
\mbox{Im}(\bs{\tau_2}) = E^- \cup \bar{U}_0 \cup \bar{U}_1 \cup
\bar{U}_2 \setminus \{ T_0^\pm, T_1^\pm, T_2^\pm \}
\end{equation}
and in particular
\begin{equation}
\vert\bs{\tau_2}^{-1}(\bt)\vert=
\begin{cases}
2 & \text{if }\ \bt \in U_0 \cup U_1 \cup U_2,\\
1 & \text{if }\ \bt \in \mbox{Im}(\bs{\tau_2}) \setminus U_0 \cup U_1 \cup U_2.
\end{cases}
\end{equation}
Furthermore, it holds:
\begin{enumerate}[(a)]
\item $\bt \in E$ if, and only if, the hyperbola branches
\begin{equation} \label{hyp-br} A_i(\bt ) = \{ \x \ \vert \ d_i(\x ) - d_0(\x ) = \tau_i
\}, i = 1, 2,\end{equation} have one of the two asymptotic lines
parallel each other. This means that there could exist a source at
a great distance from the microphones, in comparison to $d_{10}$
and $d_{20}$. If $ \bs{\tau} \in E\cap\mbox{Im}(\bs{\tau_2})$ the
hyperbola branches meet at a point at finite distance, too, which
corresponds to another admissible source position. \item $\bt\in
E^-$ if, and only if, the hyperbola branches $ A-1(\bt ) $ and $
A_2(\bt ) $ meet at one simple point and so, for a given $\bt$,
there exists a unique source position $\x$. In this case the
localization is still possible even in a noisy scenario, but we
lose in precision and stability as $ \bt $ gets close to $E$.
\item $\bt\in U_0 \cup U_1 \cup U_2$ if, and only if, the previous
hyperbola branches intersect at two distinct points, which means
that for a given $\bt$ there are two admissible source positions.
The two solutions overlap if $\bt\in\partial P_2$, which
corresponds to the tangential intersection of the branches.
\end{enumerate}

Each preimage $\x$ of a given $ \bs{\tau} \in
\mbox{Im}(\bs{\tau_2})$, i.e. the admissible source locations, is
given in a very compact form through the formalism of exterior
algebra over $ 3$--dimensional Minkowski space--time as:
\begin{equation}\label{eq:inv-image}
\mb{x}(\bs{\tau}) = \mb{L_0}(\bs{\tau}) +
\lambda(\bt) \ast((\tau_2 \mb{d_{10}} - \tau_1 \mb{d_{20}})\wedge \mb{e_3}),
\end{equation}
where $ \lambda $ is one of the real negative solutions of a
certain quadratic equation (see \cite{nostro} for the full
details). Roughly speaking, in the physical plane there are two
different regions: the preimage of the interior of the ellipse
$E^-$, where the TDOA map is $1$--to--$1$ and the source
localization is possible, and the preimage of the three disjoint regions
$U_i,\ i=0,1,2$, where the map is $2$--to--$1$ and there is no way
to uniquely locate the source. By definition, the region of transition is
exactly the bifurcation curve $\tilde{E}$, that can be
characterized as the inverse image of the ellipse $E$. In the
remaining part of this section we recall the main results on the
behavior of $\bs{\tau_2}$ in the physical plane. In Figure
\ref{fig:x-plane} we give two examples of bifurcation curve and
the relative sets.
\begin{figure}[htb]
\begin{center}
\resizebox{11cm}{!}{
  \includegraphics%[bb=-0 -0 1654 667]
  {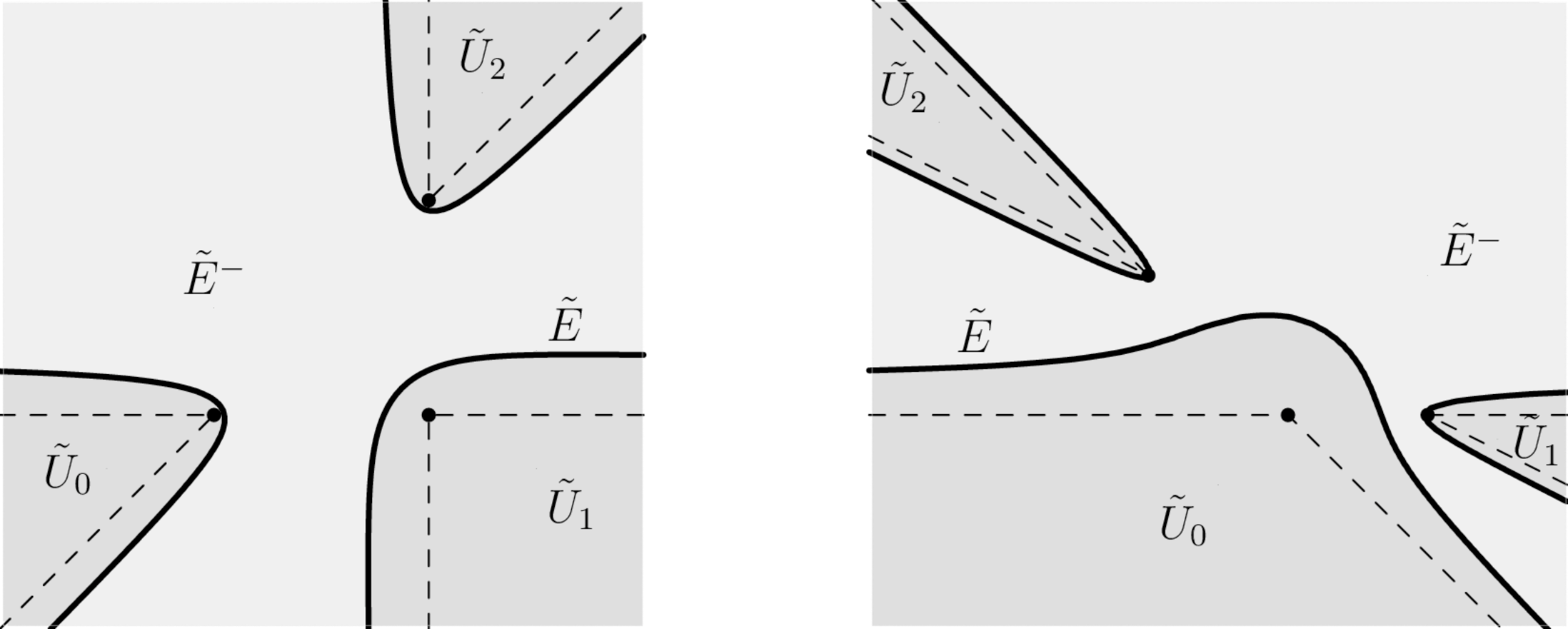}}
\caption{\label{fig:x-plane}Two examples of the different localization regions and the curve $ \tilde{E} $ in the $x$--plane. The microphones are the marked points $ \m{0}=(0,0),\ \m{1}=(2,0),$ and either $ \m{2}=(2,2) $ on the left, or $ \m{2}=(-2,2) $ on the right. Each curve $\tilde{E}$ separates the light gray region $\tilde{E}^-$, where the map $\bs{\tau_2}$ is 1--1 and it is possible to locate the source, and the medium gray region $\tilde{U}_0\cup\tilde{U}_1\cup\tilde{U}_2$, where $\bs{\tau_2}$ is 2--1 and the localization is not unique. On the dashed lines the localization is possible but unstable.}
\end{center}
\end{figure}
\begin{enumerate}[(a)]
\item As we said above, $\tilde{E} =\bs{\tau_2}^{-1}(E)$. If $\bt\in E$,
$\lambda(\bt)$ is a rational function and \eqref{eq:inv-image}
becomes a rational parametrization of $\tilde{E}$. In
\cite{nostro} it has been proved that $\tilde{E}$ is a rational
algebraic curve of degree $5$, singular on the complex plane but
smooth on the real one. \item The real part of $ \tilde{E} $
consists of three disjoint and unbounded arcs, one for each arc of
$ E $ contained in $ \mbox{Im}(\bs{\tau_2}).$ In particular, when
$\bt $ gets close to one point among the $T_i^\pm$'s in $ E \cap C
$, the point $\x$ on $ \tilde{E} $ goes to infinity. This fact and
the invariance of \eqref{eq:inv-image} under the reflection
$\bt\rightarrow -\bt$ imply that the quintic $\tilde{E}$ has three
real ideal points (the ones of the lines $ r_0, r_1, r_2,$), while
the remaining two ones are the (complex)  preimages of the
(complex) ideal points of $ E.$ Finally, the points $ \m{0},
\m{1}, \m{2} $ do not belong to $\tilde{E},$ because their images
via $ \bs{\tau_2} $ are not on $ E.$ \item Let $ \tilde{U}_i $ be
the inverse image of $U_i $ via $ \bs{\tau_2},$ for $ i=0,1,2,$
and $ \tilde{E}^- $ be the inverse image of $ E^-.$ Then,
$\tilde{E}^-, \tilde{U}_0, \tilde{U}_1, \tilde{U}_2 $ are open
subsets of the $ x$--plane, which are separated by the three arcs
of $ \tilde{E} $. On $\tilde{E}^-$ the TDOA map is $1$--to--$1$,
while it is $2$--to--$1$ on each $\tilde{U}_i, \ i=0,1,2$.
Moreover, the dashed half--lines in Figure \ref{fig:x-plane}
outgoing from the receivers divide each $ \tilde{U}_i $ into two
connected components and $ \bs{\tau_2} $ is $ 1$--to--$1 $ on each
of them. \item The source localization is possible if $\bt\in E^-$
and consequently $\x\in \tilde E^-$. Otherwise, assume $ \bs{\tau}
\in U_i$, then there are two admissible sources in the two
disjoint components of $\tilde{U}_i$. As $\bs{\tau}$ gets close to
$E$, then one of its inverse images gets close to a point on
$\tilde{E}$, while the other one goes to infinity. Conversely, if
$ \bs{\tau} $ gets close to $\partial P_2$, the inverse images of
$\bs{\tau}$ come close to each other and get close to a point on
one of the dashed half--lines. As we said before, in a realistic
noisy scenario, we have poor localization in the region close to
$\tilde{E}$.
\end{enumerate}

\section{The algebraic equation of $\tilde{E}$ and its asymptotic lines}

In this section, we use the formalism of exterior algebra over the
Euclidean plane, and we refer to Appendix A of \cite{nostro}) for
a brief introduction and summary on the subject and for the
notation. However, from the general results, it follows that
$$
\ast (\bold{u} \wedge \bold{v} ) = \det \left( \begin{array}{cc}
u_1 & v_1 \\ u_2 & v_2 \end{array} \right) $$ where $ \bold{u} =
u_1 \bold{e}_1 + u_2 \bold{e}_2, \bold{v} = v_1 \bold{e}_1 + v_2
\bold{e}_2 $ and $ (\bold{e}_1, \bold{e}_2 ) $ is an orthonormal
basis of $ \RR^2.$
\begin{definition}\label{def:1}
Let us define:
\begin{enumerate}[(a)]
\item
$
\mb{D_0}(\x)=d_0(\x) \mb{d_{12}},\quad
\mb{D_1}(\x)=d_1(\x) \mb{d_{20}},\quad
\mb{D_2}(\x)=d_2(\x) \mb{d_{01}};
$
\item
$D_i(\x)=\Vert\mb{D_i(\x)}\Vert\quad and\quad
\displaystyle{p_i(\x)=\frac{\mb{D_i}(\x)\cdot \mb{d_0}(\x)}{d_i(\x)}},\ i=0,1,2;$
\item
$W=\ast(\mb{d_{10}} \wedge \mb{d_{20}})$;
\item
$
Q(\x)=D_0(\x)^2+D_1(\x)^2+D_2(\x)^2-W^2;
$
\item
$
P_{ij}(\x)=\mb{D_i}(\x)\!\cdot\!\mb{D_j}(\x)\quad and\quad
\displaystyle{p_{ij}=\frac{P_{ij}(\x)}{d_i(\x)d_j(\x)}},\quad i,j=0,1,2.
$
\end{enumerate}
\end{definition}

\begin{theorem}
An algebraic equation for the quintic curve $\tilde{E}$ is $F(\x)=0$, where:
\begin{equation}\label{eq:CartBif7}
\begin{array}{l}
F(\x)=Q(\x)^4-8\,Q(\x)^2(P_{01}(\x)^2+P_{12}(\x)^2+P_{20}(\x)^2)\,+\\
64\,Q(\x)P_{01}(\x)P_{12}(\x)P_{20}(\x)+16\,(P_{01}(\x)^4+P_{12}(\x)^4+P_{20}(\x)^4)\\
-32\,(P_{01}(\x)^2P_{12}(\x)^2+P_{12}(\x)^2P_{20}(\x)^2+P_{20}(\x)^2P_{01}(\x)^2).
\end{array}
\end{equation}
The polynomial $F(\x)$ is invariant under permutation of the points $\m{0},$
$\m{1},\m{2}$. Expanding \eqref{eq:CartBif7} with respect to $d_0(\x)$, we have:
\begin{equation}\label{eq:CartBif8}
\begin{array}{ll}
F(\x)\!=\!\!\!\!&((W^2+2(d_{20}^2\,p_2(\x)-d_{01}^2\,p_1(\x)))^2+ 16\,p_{12}^2\,p_1(\x)\,p_2(\x))^2+\\
&8\,d_0(\x)^2(-W^4\,p_{12}\,(-8\, d_{01}^2 d_{20}^2+p_{01}\,p_{20}+2W^2)+\\
&2W^2((W^2+2\,p_{12}^2)(W^2+3\,p_{12}^2)(p_2(\x)-p_1(\x))+\\
&3\,p_{12}\,(2\,d_{01}^2\,d_{20}^2+W^2)(d_{20}^2\,p_2(\x)-d_{01}^2\,p_1(\x))+\\
&2\,p_{12}^2\,(p_{01}\,d_{20}^2\,p_2(\x)-p_{20}\,d_{01}^2\,p_1(\x)))+\\
&4(d_{20}^2\,p_2(\x)-d_{01}^2\,p_1(\x))^2
(2(d_{01}^2+d_{20}^2)W^2+7\,p_{01}\,p_{12}\,p_{20})-\\
&8(d_{20}^2\,p_2(\x)-d_{01}^2\,p_1(\x))(W^2p_{12}^2\,(p_2(\x)-p_1(\x))+\\
&(W^2+4p_{12}^2)(p_{01}^2\,p_2(\x)-p_{20}^2\,p_1(\x)))-\\
&16\,p_{12}\,p_1(\x)p_2(\x)(2W^2d_{01}^2d_{20}^2+p_{12}^2\,p_{20}\,p_{01})+\\
&2\,d_0(\x)^2(p_{12}(d_{01}^2+d_{20}^2)(4\,d_{01}^2d_{20}^2W^2+\\
&p_{12}^2(2\,p_{20}\,p_{01}-p_{12}(d_{01}^2+d_{20}^2)))
+d_{01}^2d_{20}^2(4d_{01}^4d_{20}^4-7W^4))-\\
&8Wd_{01}^2d_{12}^2d_{20}^2\!\ast\!\!(\mb{d_{01}}\wedge\mb{d_0}(\x))
\!\ast\!\!(\mb{d_{12}}\wedge\mb{d_0}(\x))
\!\ast\!\!(\mb{d_{20}}\wedge\mb{d_0}(\x))).
\end{array}
\end{equation}
\end{theorem}
\begin{proof}
The bifurcation curve $\tilde{E}$ is the preimage of the ellipse
$E$, therefore we obtain a Cartesian equation of $\tilde{E}$ by
substituting $(\tau_1,\tau_2)=(\tau_{10}(\x),\tau_{20}(\x))$ in
\eqref{eq:a}:
\begin{equation}\label{eq:CartBif}
\Vert (d_2(\x)-d_0(\x)) \mb{d_{10}} - (d_1(\x)-d_0(\x)) \mb{d_{20}} \Vert^2 =
(\ast (\mb{d_{10}} \wedge \mb{d_{20}}) )^2.
\end{equation}
Using Definitions \ref{def:1}, we have the more symmetric form
\begin{equation}\label{eq:CartBif2}
\Vert \mb{D_0}(\x) + \mb{D_1}(\x) + \mb{D_2}(\x) \Vert^2 = W^2
\end{equation}
and, after expanding the left hand side,
\begin{equation}\label{eq:CartBif3}
\begin{array}{l}
D_0(\x)^2+D_1(\x)^2+D_2(\x)^2+\\
2\,(\mb{D_0}(\x)\!\cdot\!\mb{D_1}(\x)+
\mb{D_1}(\x)\!\cdot\!\mb{D_2}(\x)+\mb{D_2}(\x)\!\cdot\!\mb{D_0}(\x))=W^2.
\end{array}
\end{equation}
Of course, this is not an algebraic equation with respect to
$x,y$. In order to obtain one, we use again Definitions
\ref{def:1} and we rewrite equation \eqref{eq:CartBif3} as
\begin{equation}\label{eq:CartBif4}
Q(\x)+2P_{12}(\x)=-2\,(P_{01}(\x)+P_{20}(\x)).
\end{equation}
By squaring both sides and reordering, we obtain:
\begin{equation}\label{eq:CartBif5}
\begin{array}{l}
Q(\x)^2-4(P_{01}(\x)^2-P_{12}(\x)^2+P_{20}(\x)^2)=\\
-4\,Q(\x)\,P_{12}(\x)+8\,P_{01}(\x)\,P_{20}(\x)
\end{array}
\end{equation}
Again, the right side of the last equation is not a polynomial,
but squaring once we get the algebraic equation:
\begin{equation}\label{eq:CartBif6}
\begin{array}{l}
(Q(\x)^2-4(P_{01}(\x)^2-P_{12}(\x)^2+P_{20}(\x)^2))^2=
16\,Q(\x)^2\,P_{12}(\x)^2\\
-64\,Q(\x)\,P_{01}(\x)\,P_{12}(\x)\,P_{20}(\x)+64\,P_{01}(\x)^2\,P_{20}(\x)^2,
\end{array}
\end{equation}
that coincides with equation $F(\x)=0$. It is straightforward to
verify that \eqref{eq:CartBif7} is invariant with respect to
permutations of the points $\m{0},\m{1},\m{2}$.

The degree of polynomial \eqref{eq:CartBif7} with respect to $(x,y)$ is 8 at the most. By calculating the Taylor expansion
\eqref{eq:CartBif8} of \eqref{eq:CartBif7} centered at the point
$\m{0}$, we show that $F(\x)$ has degree $5$ and this completes
the proof (the verification of expansion \eqref{eq:CartBif8} is a
simple matter of computation).
\end{proof}

\begin{remark}\rm
If the receivers are not collinear, we have that
$F(\m{0})=F(\m{1})=F(\m{2})=W^8>0$. Therefore $\x\in \tilde{E}^-$
if, and only if, $F(\x)<0$ and $\x\in \tilde{E}^+$ if, and only
if, $F(\x)>0$.
\end{remark}

As an example, we provide the Cartesian equations of the two bifurcation curves of Fig. \ref{fig:x-plane}. The bifurcation curve on the left has equation
\begin{equation*} \begin{split} \tilde{E}: &
-4x^4y+4x^3y^2-4x^2y^3+4xy^4+2x^4+20x^3y-16x^2y^2+4xy^3 + \\ & -6y^4-10x^3-38x^2y+30xy^2+2y^3+18x^2+28xy-22y^2+ \\ & - 12x-4y+1=0
\end{split}
\end{equation*}
while the one of the right has equation
\begin{equation*} \begin{split} \tilde{E}: &
-20x^4y-60x^3y^2-60x^2y^3-60xy^4-40y^5+10x^4+68x^3y+ \\ & +80x^2y^2+84xy^3+82y^4-34x^3-58x^2y-10xy^2-50y^3+30x^2 + \\ & +4xy+22y^2-4x-4y+1=0.
\end{split}
\end{equation*}

Using polynomial \eqref{eq:CartBif8} we can compute an algebraic
expression for the real asymptotic lines of $\tilde{E}$. We refer
to Appendix B of \cite{nostro} for an introduction to projective
geometry. As a preliminary, we prove the following Lemma.
\begin{lemma}\label{Lemma}
$\displaystyle{W=\frac{\ast(\mb{d_{12}}\wedge\mb{d_0}(\x))+ \ast(\mb{d_{20}}\wedge\mb{d_1}(\x))+ \ast(\mb{d_{01}}\wedge\mb{d_2}(\x))}{2}}$.
\end{lemma}
\begin{proof}
We use the following identities
\begin{equation}
\mb{d_{01}}+\mb{d_{12}}+\mb{d_{20}}=0,\qquad
\mb{d_{ij}}=\mb{d_j}(\x)-\mb{d_i}(\x),\qquad i,j=0,1,2.
\end{equation}
We have
\begin{equation}
\begin{array}{ll}
\mb{d_{10}}\wedge\mb{d_{20}}
&=(\mb{d_{0}}(\x)-\mb{d_{1}}(\x))\wedge\mb{d_{20}}=\\
&=\mb{d_0}(\x)\wedge(\mb{d_{21}}+\mb{d_{10}})-\mb{d_1}(\x)\wedge\mb{d_{20}}=\\
&=\mb{d_{12}}\wedge\mb{d_0}(\x)+\mb{d_{20}}\wedge\mb{d_1}(\x)
+(\mb{d_2}(\x)+\mb{d_{20}})\wedge\mb{d_{10}}=\\
&=\mb{d_{12}}\wedge\mb{d_0}(\x)+\mb{d_{20}}\wedge\mb{d_1}(\x)
+\mb{d_{01}}\wedge\mb{d_2}(\x)-\mb{d_{10}}\wedge\mb{d_{20}}.
\end{array}
\end{equation}
The Lemma follows from Definition \ref{def:1} and the last identity.
\end{proof}
\begin{definition}
Let $ \AF^2 $ be the affine plane, and let $ \PP^2 $ be the
projective plane obtained by joining the ideal line $ \ell $ to $
\AF^2.$ Let $C$ be an algebraic curve in the affine plane $\AF^2$.
The ideal points of $C$ are the intersection points of $C$ and the
ideal line $ \ell.$ An asymptotic line of $C$ is a line in $ \AF^2
$ tangent to $C$ at one of its smooth ideal points.
\end{definition}

Let $ f(x,y) = 0 $ be the Cartesian equation of a degree $ d $
algebraic curve $ C $ in $ \AF^2.$ We can write it as $$ f(x,y) =
f_d(x,y) + f_{d-1}(x,y) + \dots + f_1(x,y) + f_0 $$ where $
f_i(x,y) $ is homogeneous of degree $ i.$ We embed $ \AF^2 $ into
$ \PP^2 $ by setting $ x = X/U, y = Y/U,$ and $ \ell: U=0 $ is the
ideal line. The curve $ C \subset \PP^2 $ is then defined by $$
F(X,Y,U)= f_d(X,Y) + f_{d-1}(X,Y) U + \dots + f_1(X,Y) U^{d-1} +
f_0 U^d.$$ The ideal points are the solutions, in the sense of
projective geometry, of $ U = f_d(X,Y) = 0.$ Let $ [a:b:0] $ be a smooth
ideal point of $ C.$ So, the line $ r: -bx+ay+c=0 $ is an
asymptotic line for $ C $ if $ [a:b:0] $ is a solution of $
F(X,Y,U) = -bX+aY+cU = 0 $ of multiplicity at least two.

\begin{theorem}
The bifurcation curve $\tilde{E}$ has three real ideal points
$[x_1-x_2:y_1-y_2:0],[x_2-x_0:y_2-y_0:0],[x_0-x_1:y_0-y_1:0]$ and
two complex ones $[1:i:0],[1:-i:0].$\\
The three real asymptotic lines of $\tilde{E}$ have Cartesian equation
\begin{equation}\label{eq:asintoti}\begin{array}{l}
L_0:\ 4\ast(\mb{d_{12}}\wedge\mb{d_0}(\x))-3W=0,\\
L_1:\ 4\ast(\mb{d_{20}}\wedge\mb{d_1}(\x))-3W=0,\\
L_2:\ 4\ast(\mb{d_{01}}\wedge\mb{d_2}(\x))-3W=0.
\end{array}\end{equation}
Finally, we have $L_0\cap L_1\cap L_2=\emptyset.$
\end{theorem}
\begin{proof}
The homogeneous degree--$5$ part of polynomial \eqref{eq:CartBif8} is
\begin{equation}\label{eq:CartBif85}
-64Wd_{01}^2d_{12}^2d_{20}^2\,d_0(\x)^2
\ast(\mb{d_{12}}\wedge\mb{d_0}(\x))
\ast(\mb{d_{20}}\wedge\mb{d_0}(\x))
\ast(\mb{d_{01}}\wedge\mb{d_0}(\x))
\end{equation}
and so it is straightforward to check that the five ideal points
of the statement are its roots, and they all are smooth for $ \tilde{E}.$ By using a package for algebraic
computations, it is easy to prove that $ L_i $ meets $ C $ at the
corresponding ideal point with multiplicity $ 2,$ and this proves
that the line $ L_i $ is asymptotic to $ C.$

The last statement follows from Lemma \ref{Lemma}. In fact, if we
sum the three polynomials \eqref{eq:asintoti} defining the
asymptotic lines, we obtain
\begin{equation}\begin{array}{l}
4\ast(\mb{d_{12}}\wedge\mb{d_0}(\x)+\mb{d_{20}}\wedge\mb{d_1}(\x)+\mb{d_{01}}\wedge\mb{d_2}(\x))-9W=-W.
\end{array}\end{equation}
Thus, the three lines do not have a common intersection point: in
fact, if there exists a common point $ \x_0,$ its coordinates
satisfy also the sum of the three equations of the asymptotic
lines. Such a sum is $ W = 0,$ and so $ \x_0 $ does not exist.
\end{proof}

As $ \tilde{E} $ is a quintic curve, each line $ L_i $ intersects
$ \tilde{E} $ at either $ 1 $ or $ 3 $ real points. Hence, also if
it is not evident from Fig. \ref{fig:x-plane}, the unbounded arcs
of $ \tilde{E} $ definitely belong to different half--planes with
respect to its asymptotic lines.

\section{Conclusion}

In this paper, we recall the state of the art on the localization of
sources in a plane from the TDOAs for the case of $ 3 $ receivers
in the same plane. Then, we focus on a problem still open in
the literature: the computation of the Cartesian algebraic equation of the
bifurcation curve $ \tilde{E},$ that is to say, of the curve in
the plane of source and receivers whose points are sources for
which the hyperbola branches (\ref{hyp-br}) have an asymptotic
line parallel each other. The knowledge of such an equation allows us
to easily solve the problem of finding points in the plane which
are close or belong to such a curve $ \tilde{E}.$ The importance
of computing the equation of $ \tilde{E} $ stems from the fact that
at its points, every localization algorithm has a poor accuracy.
The computation of the Cartesian equation of the bifurcation curve
rests on two steps: first, by squaring a non--polynomial
equation for $ \tilde{E},$ we compute a polynomial; then, by using
a Taylor expansion, we get the explicit degree five equation.
From such equation, it is possible to compute the real asymptotic lines of $ \tilde{E}.$ Notice that it is not possible to compute such lines from the parametric equation of $ \tilde{E}.$

In \cite{nostro} as well as in the present paper, we completed the
geometric study of the noiseless model of localization, encoded in
the TDOA map $ \bt_2,$ for the case of $ 3 $ receivers in a plane.
In a manuscript in preparation, we will conduct a similar
study for a real scenario. Other then the deterministic case, the
needed techniques come from information geometry \cite{Amari2000} and statistics
with algebraic tools \cite{Draisma2013,Kobayashi2013}, and from numerical analysis. Moreover, still
in preparation, we are studying the geometry of the noiseless
model in the case of $ 4 $ or more receivers in a plane.

%%%%%%%%%%%%%%%%%%%%%%%%%%%%%%%%%%%%%%%%%%%%%%%%%%%%%%%%%%%%%%%%%%%%%%%%%%%%%%

%\section*{References}
\bibliographystyle{plain}   % (uses file "plain.bst")
\bibliography{biblio}       % expects file "myrefs.bib"

\end{document}